\documentclass[a4paper,11pt]{article}
\usepackage{amsmath,graphicx,amssymb,color,amsthm}
\newtheorem{theorem}{Theorem}
\newtheorem{lemma}{Lemma}

\usepackage{epsfig,subfigure,float}
\usepackage{epsfig,url}
\usepackage{color}
\usepackage[colorlinks=true,citecolor=blue, linkcolor=blue, urlcolor=magenta]{hyperref}
\newtheorem{definition}{Definition}
 \numberwithin{equation}{section}

\begin{document}

\centerline{\LARGE \bf
An age-distributed immuno-epidemiological model }

\medskip

\centerline{\LARGE \bf with information-based vaccination decision}

\medskip

\vspace*{1cm}

\centerline{\bf Samiran Ghosh$^1$, Malay Banerjee$^{1,*}$, Vitaly Volpert$^{2,3}$}

\vspace{0.5cm}

\centerline{ $^1$ Indian Institute of Technology Kanpur, Kanpur - 208016, India}

\centerline{$^2$
Institut Camille Jordan, UMR 5208 CNRS, University Lyon 1, 69622 Villeurbanne, France}

\centerline{ $^3$ Peoples Friendship University of Russia (RUDN University), 6 Miklukho-Maklaya St,} 
\centerline{ Moscow 117198, Russian Federation}

\centerline{ $^*$ Corresponding author}

\vspace{2cm}

\noindent
{\bf Abstract.} 
 A new age-distributed immuno-epidemiological model with information-based vaccine uptake suggested in this work represents a system of integro-differential equations for the numbers of susceptible individuals, infected individuals, vaccinated individuals and recovered individuals.
This model describes the influence of vaccination decision on epidemic progression in different age groups. We prove the existence and uniqueness of a positive solution using the fixed point theory. In a particular case of age-independent model, we determine the final size of epidemic, that is, the limiting number of susceptible individuals at asymptotically large time. Numerical simulations show that the information-based vaccine acceptance can significantly influence the epidemic progression. Though the initial stage of epidemic progression is the same for all memory kernels, as the epidemic progresses and more information about the disease becomes available, further epidemic progression strongly depends on the memory effect. Short-range memory kernel appears to be more effective in restraining the epidemic outbreaks because it allows for more responsive and adaptive vaccination decisions based on the most recent information about the disease.

\vspace{1cm}

\noindent
{\bf Keywords:} Epidemic model; distributed delay; age-distributed model; information-based vaccination decision; fixed point theorem

\medskip

\vspace*{0.5cm}

\setcounter{equation}{0}
\setcounter{section}{0}
\setcounter{page}{1}

\section{Introduction}
Modern epidemiology disposes a wide spectrum of mathematical models in epidemiology, mainly based upon the ordinary differential equations and partial differential equations, integro-differential equations, individual based methods, and some other modelling tools. 
The ordinary differential equation epidemic models were introduced by D. Bernoulli in the XVIII-th century in studying the smallpox epidemic \cite{first_ode}. Another remarkable step in the development of mathematical epidemiology is due to the  works by W. O. Kermack and A. G. McKendrick \cite{kermack1927contribution,kermack1932contributions,kermack1933contributions} where they established the foundation of modern developments in mathematical epidemiology. These approaches were further developed in various multi-compartmental models (see, e.g., \cite{brauer_compartmental,martcheva,mbe}), agent-based models \cite{anass2,agent_bases_model_1}, immuno-epidemic models \cite{mmnp_imep,gilchrist2002modeling,temporary_immunity_2}, network models \cite{network_model_1,network_model_2}, multi-scale models \cite{multi_scale_2,multi_scale_1}, age-structured models \cite{inaba1990threshold,book_age_structured_1}, and many other approaches.

One of the directions of these investigations concerns the influence of vaccination on epidemic progression.
Vaccination is one of the most effective measures to restrain epidemic progression. However, success of vaccination strategy depends on the available information about the epidemic and the risk of adverse events \cite{alberto_book,wang2016statistical}. COVID-19 epidemic has provided numerous examples of uncoordinated behavioral responses to infections threats and control measures \cite{perra2021non}. Vaccine hesitancy has currently been included by WHO among the most serious threats to global health \cite{macdonald2015vaccine}.


Behavioral epidemiology is an interdisciplinary field that seeks to understand how human behavior influences the spread of infectious diseases \cite{alberto_book}. It combines traditional epidemiological models with theories and methods from behavioral sciences to better understand the complex social and psychological factors that contribute to the transmission of infectious diseases. By incorporating insights from multiple disciplines, behavioral epidemiology provides a more nuanced understanding of the influence of human behavior on the infectious disease transmission. This can inform the development of more effective public health interventions that take into account the social and psychological factors that influence people's behavior \cite{bavel2020using,pearlin1989sociological}.

Epidemic models with vaccination decision based upon information were introduced in \cite{alberto_book}. One of possible developments of this theory consists in 
the introduction of age-distributed vaccination decision in the model. Different age groups may have different attitudes towards vaccination, and these attitudes may be influenced by various age specific factors such as the severity of the disease, the perceived benefits and risks of vaccination, and the cost of vaccination.
In the case of influenza, the  distribution of vaccinated people according to the age groups is analyzed in \cite{buchan2016trends}. Recent studies of COVID-19 epidemic reveal that younger people are more hesitant to vaccination than older people \cite{ahamed2021understanding,cordina2021attitudes,neumann2020once}.
This is clearly related to the age-distributed death rate for the SARS-CoV-2 infection. Numerous studies  for different infectious diseases show that the vaccination decision varies between the age groups. Thus, it is very relevant to study the impact of age specific vaccination decision on the epidemic outbreaks.

In this work we introduce and study an age-distributed epidemic model with information-based vaccination decision. This study continues the previous works on epidemic models that involves time-since-infection distributed recovery rates and death rates \cite{mmnp_imep,bmb,age_dependent_our1}. This approach allows us to take into account time-dependent recovery and death rates which give a more accurate understanding of epidemic progression than conventional SIR models. Here we extend the previous models with information-based vaccination.

The paper is organized as follows. We propose the age-distributed model with information-based vaccination decision  in Section \ref{model_formulation}. In Section \ref{existence_uniqueness_1}, we prove the existence and uniqueness of positive solution of the proposed model. Section \ref{age_independent_model} is devoted to the reduced  age-independent model for which we determine an analytical bound for the final size of epidemic depending on the information-based vaccination uptake. The influence of information-dependent vaccination uptake on the epidemic progression is illustrated with the help of numerical simulation in Section \ref{simulation_results}.




\setcounter{equation}{0}

\section{Model formulation}\label{model_formulation}

\subsection{Distributed recovery and death rates}

Let individual's age $a$ vary between $0$ and $A^{\dagger}$, and $J(a,t)$ be a non-negative function denoting the number of individuals of age $a$ at time $t$, who are newly infected. Suppose that $S(a,t)$, $I(a,t)$, $R(a,t)$, and $D(a,t)$ represent the numbers of susceptible individuals, infected individuals, recovered individuals and dead individuals of age $a$ at time $t$. Then the cumulative number of infected individuals of age $a$ at time $t$  (i.e., $\int_0^t J(a,\zeta) d \zeta$) can be written as follows \cite{bmb}:
\begin{equation}\label{s0}
\int_0^t J(a,\zeta) d \zeta =I(a,t)+R(a,t) + D(a,t).
\end{equation}
We also suppose that for each age-group $a$, the total population is constant for all time $t>0$, i.e., $$(S+I+R+D)(a,t)=S_0(a),$$ where $S_0(a)$ is the initial population size in the age-group $a$. Using this assumption and differentiating (\ref{s0})  with respect to $t$, we obtain:
\begin{equation}
    \label{s1}
    \frac{\partial S(a,t)}{\partial t} = - J(a,t).
\end{equation}
Next, we suppose that the infection rate varies proportionally to the total viral load in the population determined by the viral load within an infected individual load and the number of infected individuals \cite{webb2005model}.
Let us denote the total viral load during the whole period of infectiousness of an individual of age $a$ by the variable $V(a)$. Then using \eqref{s1} we get
\begin{equation}
    \label{s2}
   J(a,t) =- \frac{\partial S(a,t)}{\partial t} = \alpha(a) S(a,t) \int_0^{A^{\dagger}} V(y) I(y,t) dy,
\end{equation}
where the coefficient $\alpha(a)$ is the age-distributed susceptibility rate. In the case of age-independent infectivity rate, the last formula becomes similar to the conventional SIR model. 

Let $r(a,\xi)$ and $d(a,\xi)$, respectively denote, the recovery and death distributions for the infected individuals in the age group $a$ depending on time-post-infection $\xi$. Then, the quantities $R_n(a,t)$ and $D_n(a,t)$, denoting the recovery and death at time $t$ in the age group $a$, respectively, are determined by the following expressions:
$$ R_n(a,t)=\int_0^t r(a,t-\zeta) J(a,\zeta) d \zeta \; , \;\;\;
 D_n(a,t)=\int_0^t d(a,t-\zeta) J(a,\zeta) d \zeta.$$
Hence, we obtain the governing equation for $I(a,t)$, $R(a,t)$ and $D(a,t)$:
\begin{equation}
    \label{s3}
    \frac{\partial I(a,t)}{\partial t}=J(a,t)   - \int_0^t r(a,t-\zeta) J(a,\zeta) d \zeta -\int_0^t d(a,t-\zeta) J(a,\zeta) d \zeta ,
\end{equation}
\begin{equation}
    \label{s4}
     \frac{\partial R(a,t)}{\partial t} = R_n(a,t)  ,\,\, \frac{\partial D(a,t)}{\partial t} = D_n(a,t) .
\end{equation}


\subsection{Information-based vaccination} 
Let $\rho(a,t)$ denote the proportion of vaccinated individuals of age $a$ at time $t$, $\rho(a,t)$ varies between $0$ and $1$. 
Assume that individuals who are vaccinated can never become infected. Then instead of equations (\ref{s1})-(\ref{s3}) we get the equations
\begin{equation}\label{pp11}
  \frac{\partial S(a,t)}{\partial t} = - \alpha(a) S(a,t)(1-\rho(a,t))\int_0^{A^{\dagger}} V(y) I(y,t) dy \;\;( \equiv -J(a,t)),
\end{equation}
\begin{equation}\label{pp12}
   \frac{\partial I(a,t)}{\partial t}= \alpha(a) S(a,t)(1-\rho(a,t))\int_0^{A^{\dagger}} V(y) I(y,t) dy   - \int_0^t r(a,t-\zeta) J(a,\zeta) d \zeta -\int_0^t d(a,t-\zeta) J(a,\zeta) d \zeta .
\end{equation}
The governing equation for the age-specific vaccination rate is considered in the form
\begin{equation}
\label{s5}
\frac{\partial \rho(a,t)}{\partial t} = \mathcal{P}(m(I(a,t))) (1-\rho(a,t)). 
\end{equation}
Here $m(a,t)$ is an information index which summarizes the age-specific and time varying information on the present and past incidence of the infection along with its sequel \cite{d2007vaccinating}, $\mathcal{P}(m)$ is a positive function that represents the probability of vaccine uptake, which takes into account the vaccine hesitancy. It can be considered as a monotonically increasing function of the information index $m$ implying that better informed people are more inclined for vaccination \cite{d2007vaccinating}.

Hence, the age-distributed epidemic model with information based vaccine uptake is given by equations (\ref{s4})-(\ref{s5})
with the initial conditions
\begin{equation}\label{main_model_IC}
S(a,0)=S_0(a),\; I(a,0)=I_0(a),\; R(a,0)=R_0(a),\; D(a,0)=D_0(a),\;\rho(a,0)=\rho_0(a).
\end{equation}
Here all the functions are continuous and non-negative for $x\in[0,A^{\dagger}]$.

Let us consider the following condition for the recovery and death rates:
\begin{equation}
\label{cond}
\int_{\zeta}^{t_0} (r(a,\xi-\zeta) + d(a,\xi-\zeta)) d\xi \leq 1
\end{equation}
for any $\zeta$ and $t_0$, $t_0 > \zeta$ and $a$. This is an epidemiologically justified condition because the left hand side represents the proportion of total recovery and death during the time $\zeta$ to $t_0$, among the individuals who were infected at time $\zeta$, and this proportion must be less than or equal to $1$ (see \cite{age_dependent_our1} for further details).


\section{Existence $\&$ uniqueness of positive solution}\label{existence_uniqueness_1}

In this section, we investigate the existence, uniqueness and positiveness of solution for the system \eqref{s4}-\eqref{s5}. We begin with some preliminary results.


\medskip
\subsection{Preliminary results}

We recall some definitions and results that will be used in proving the existence and uniqueness results.

\begin{definition} 
Let $f:S_1 \to S_2$ be a single-valued mapping, where $S_1$ and $S_2$ are two non-empty sets. If for some $x\in S_1$, $f(x)=x$, then $x$ is called a fixed point of $f$.
\end{definition}
\medskip
\begin{definition} \cite{banach1922operations}
Let $(M,d)$ be a metric space and $f:(M,d)\to (M,d)$ is a mapping. $f$ is called a Banach contraction mapping if there exists a real number $r \in [0, 1)$ such that,
$$d(f(x),f(y)) \leq r\; d(x,y),$$
for all $x,y \in M.$
\end{definition}

\medskip

\begin{theorem}\label{Banach_theorem} \cite{banach1922operations}
Let $(M,d)$ be a complete metric space and $f:(M,d) \to (M,d)$ be a Banach contraction mapping. Then $f$ has a unique fixed point in $x$.
\end{theorem}

Before the existence and uniqueness theorem, we prove the positiveness of solution for the system \eqref{s4}-\eqref{s5} in the following subsection.


\subsection{Positiveness}

\begin{lemma}
If condition (\ref{cond}) holds, then the solutions $S(a,t)$, $I(a,t)$, $R(a,t)$, and $D(a,t)$ of the model \eqref{s4}-\eqref{main_model_IC} is bounded by $0$ and $S_0(a)+I_0(a)$. 
\end{lemma}

\begin{proof}
First of all, from \eqref{s5} we note that, for some $x=x_1$ and $t=t_1$, $\rho(a_1,t_1)=1$ implies $\left.\frac{\partial \rho}{\partial t}\right|_{(a_1,t_1)}=0$. Hence, $0\leq \rho(a,t)\leq 1$, for all $x,t$.
From (\ref{pp11}), clearly $\left.\frac{\partial S(a,t)}{\partial t}\right|_{t=t^*} = 0$, provided $S(a,t^*) =0$ for some $a$ and $t^*$. Hence $S(a,t)$ is always greater than or equal to $0$ for all $a$, $t >0$. Also from (\ref{s4}) it is clear that $R(a,t)$ and $D(a,t)$ are positive for all $a$, $t$. 

Also we have,
\begin{equation}\label{positivity}
 I(a,t_0) =   \int_0^{t_0} J(a,\zeta) d \zeta - R(a,t_0) - D(a,t_0) ,
\end{equation}
for some $t=t_0 > 0$. With the initial conditions $R(a,0)=D(a,0)=0$, we integrate (\ref{s4}) from $0$ to $t_0$ with respect to $t$
\begin{equation*}
R(a,t_0) + D(a,t_0) = \int_0^{t_0} \bigg( \int_0^{t} \big( r(a,t-\zeta) + d(a,t-\zeta) \big) J(a,\zeta) d \zeta \bigg) dt.
\end{equation*}
Changing the order of integration and using the assumption (\ref{cond}), we get the following
\begin{subequations}
\begin{eqnarray*}
R(a,t_0) + D(a,t_0) &=& \int_0^{t_0} \bigg( \int_{\zeta}^{t_0} \big( r(a,t-\zeta) + d(a,t-\zeta) \big) dt \bigg) J(a,\zeta)  d \zeta \leq \int_0^{t_0} J(a,\zeta)  d\zeta.
\end{eqnarray*}
\end{subequations}
Now using the condition (\ref{positivity}), we obtain $I(a,t_0) \geq  0$. Moreover,
$$S(a,t)+I(a,t)+R(a,t)+D(a,t)=S_0(a)+I_0(a).$$
Thus, any solution of system \eqref{s4}-\eqref{s5} lies between $0$ and $S_0(a)+I_0(a)$.
\end{proof}


\subsection{Existence of unique positive solution}

We prove the existence of unique positive solution of the system \eqref{s4}-\eqref{s5} for $(a, t) \in [0, A^{\dagger}] \times [0, \mathcal{T}]$, where $A^{\dagger}, \mathcal{T} \in (0, \infty)$. 

 \begin{theorem}\label{existence_uniqueness}
If  $ \mathcal{P}(m)$ is a Lipschitz continuous function of $m$, then there exists a unique solution of system \eqref{s4}-\eqref{s5} in the domain $\mathcal{M}^4$, where
 $$\mathcal{M} = \bigg\{\Phi \in C\big( [0, A^{\dagger}] \times [0, \mathcal{T}],\; \mathbb{R} \big) \; : \; 0 \leq \Phi(a,t) \leq S_0(a) + I_0(a),\;\; \forall (a,t) \in [0, A^{\dagger}] \times [0, \mathcal{T}] \bigg\}.$$
 \end{theorem}

 
 To prove this theorem we use the complete metric space setting  $\big( \mathcal{M}, d \big) $ (proof of completeness is given in Appendix 1) and the metric $d$ defined by
 $$d(\Phi_1,\Phi_2)=\sup_{(a,t)\in [0,A^{\dagger}]\times[0,\mathcal{T}]} \bigg\{ e^{-\delta t} | \Phi_1(a,t)-\Phi_2(a,t)| \bigg\},$$
 where $\delta >0$ is a constant. 
 For any given function $\Phi \in \mathcal{M} $, we set
 
  \begin{equation}\label{eqn_mt}
m_{\Phi}(a,t)=m(\Phi(a,t)).
\end{equation}
Then for this choice of $m_{\Phi}(a,t)$ and $\mathcal{P}(m_{\Phi}(a,t))$, the equation

\begin{equation}\label{eqn_rho}
   \frac{\partial \rho(a,t)}{\partial t} = \mathcal{P}(m_{\Phi}(a,t)) (1-\rho(a,t)),
 \end{equation}
 with $\rho(a,0)=\rho_0(a)$, has a unique solution satisfying the equation
\begin{equation}\label{eqn_rho_2}
 (1-\rho_{\Phi}(a,t)) = (1-\rho_0(a)) e^{- \int_0^t \mathcal{P}(m_{\Phi}(a,\zeta)) d \zeta}.
\end{equation}

For this given choice of $\Phi(a,t)$ and $\rho_{\Phi}(a,t)$, the equation
 \begin{equation}\label{SI_3}
\frac{\partial S(a,t)}{\partial t} = - \alpha(a) S(a,t)(1-\rho_{\Phi}(a,t)) \int_0^ {A^{\dagger}} V(y) \Phi(y,t) dy,
\end{equation}
with $S(a,0)=S_0(a)$ has a unique solution (since the right-hand side is a Lipschitz continuous function in $S$) and the solution can be written as follows:
\begin{equation}\label{SI_4}
 S_{\Phi}(a,t) = S_0(a) e^{- \alpha(a) \int_0^t (1-\rho_{\Phi}(a,\zeta))\big( \int_0^ {A^{\dagger}}  V(y) \Phi(y,\zeta) dy \big) d \zeta} .
\end{equation}
Let
$$J_{\Phi}(a,t)= \alpha(a) (1-\rho_{\Phi}(a,t)) S_{\Phi}(a,t) \int_0^ {A^{\dagger}} V(y) {\Phi}(y,t) dy.$$
Then the equation

\begin{equation}
 \frac{\partial I(a,t)}{\partial t} =  \alpha(a) (1-\rho_{\Phi}(a,t)) S_{\Phi}(a,t) \int_0^ {A^{\dagger}} V(y) \Phi(y,t) dy - \int_0^t \big( r(a,t-\zeta) + d(a,t-\zeta) \big)  J_{\Phi}(a,\zeta) d \zeta,
\end{equation}
with $I(a,0)=I_0(a)$ has unique solution which can be written in the form

\begin{eqnarray}\label{eqn_i1}
I_{\Phi}(a,t) &=& I_0(a) + \int_0^t \mathcal{A}(a,\zeta, \Phi) d \zeta,
\end{eqnarray}
where
\begin{eqnarray}\label{eqn_G}
\mathcal{A}(a,\zeta, \Phi) &=&  \alpha(a) (1-\rho_{\Phi}(a,\zeta)) S_0(a) e^{- \alpha(a) \int_0^{\zeta} (1-\rho_{\Phi}(a,\xi)) \big( \int_0^ {A^{\dagger}}  V(y) \Phi(y,\xi) dy \big) d \xi} \int_0^ {A^{\dagger}} V(y) \Phi(y,\zeta) dy \nonumber \\
 && - \int_0^{\zeta} \bigg[ \big( r(a,\zeta-\xi) + d(a,\zeta-\xi) \big)\alpha(a) (1-\rho_{\Phi}(a,\xi)) \nonumber \\
 && S_0(a) e^{- \alpha(a) \int_0^{\xi} (1-\rho_{\Phi}(a,\nu))\big( \int_0^ {A^{\dagger}}  V(y) \Phi(y,\nu) dy \big) d \nu} \int_0^ {A^{\dagger}}  V(y) \Phi(y,\xi) dy \bigg]  d \xi .
\end{eqnarray}

For further proof, we need the following lemma.

\begin{lemma}\label{define_L}
The map $\mathcal{F}:\; (\mathcal{M},d) \to (\mathcal{M},d)$ defined by the eqality
$$\mathcal{F}(\Phi)(a,t) = I_0(a) + \int_0^t \mathcal{A}(a,\zeta, \Phi) d \zeta,$$
where $\mathcal{A}(a,\zeta,\Phi)$ is given in \eqref{eqn_G}, is well-defined.
\end{lemma}

\begin{proof}
From, \eqref{SI_4}, we observe that
$$\frac{\partial S_{\Phi}(a,t)}{\partial t} = -\alpha(a) (1-\rho_{\Phi}(a,t)) S_0(a) e^{- \alpha(a) \int_0^{t} (1-\rho_{\Phi}(a,\xi)) \big( \int_0^ {A^{\dagger}}  V(y) \Phi(y,\xi) dy \big) d \xi} \int_0^ {A^{\dagger}} V(y) \Phi(y,t) dy . $$
Using this relation in \eqref{eqn_G}, we can write,
$$\mathcal{A}(a,\zeta, \Phi) = - \bigg[ \frac{\partial S_{\Phi}(a,\zeta)}{\partial \zeta} - \int_0^\zeta \big( r(a,\zeta-\xi) + d(a,\zeta-\xi) \big) \frac{\partial S_{\Phi}(a,\xi)}{\partial \xi} d \xi \bigg] ,$$
and
\begin{subequations}
\begin{eqnarray*}
\int_0^t \mathcal{A}(a,\zeta, \Phi) d \zeta &=& - \bigg[ \int_0^t \frac{\partial S_{\Phi}(a,\zeta)}{\partial \zeta} d \zeta - \int_0^t \int_0^{\zeta} \big( r(a,\zeta-\xi) + d(a,\zeta-\xi) \big) \frac{\partial S_{\Phi}(a,\xi)}{\partial \xi} d \xi d \zeta \bigg].
\end{eqnarray*}
\end{subequations}
We change the order of integration to get the following
\begin{eqnarray*}
\int_0^t \mathcal{A}(a,\zeta, \Phi) d \zeta &=& - \bigg[ \int_0^t \frac{\partial S_{\Phi}(a,\zeta)}{\partial \zeta} d \zeta - \int_0^t \bigg( \int_{\xi}^{t} \big( r(a,\zeta-\xi) + d(a,\zeta-\xi) \big) d \zeta \bigg)  \frac{\partial S_{\Phi}(a,\xi)}{\partial \xi} d \xi \bigg] \\
&=& - \bigg[ \int_0^t \bigg(1- \int_{\xi}^{t} \big( r(a,\zeta-\xi) + d(a,\zeta-\xi) \big) d \zeta \bigg)  \frac{\partial S_{\Phi}(a,\xi)}{\partial \xi} d \xi \bigg].
\end{eqnarray*}
Note that $\frac{\partial S_{\Phi}(a,\xi)}{\partial \xi} \leq 0$, and using the condition \eqref{cond}, we can write,
\begin{eqnarray*}
\int_0^t \mathcal{A}(a,\zeta, \Phi) d \zeta & \leq & -  \int_0^t  \frac{\partial S_{\Phi}(a,\xi)}{\partial \xi} d \xi = S_0(a)-S_{\Phi}(a,t) \; \leq \; S_0(a).
\end{eqnarray*}
This implies the estimate 
\begin{eqnarray*}
\mathcal{F}(\Phi)(a,t) &=& I_0(a) + \int_0^t \mathcal{A}(a,\zeta, \Phi) d \zeta \leq I_0(a) + S_0(a).
\end{eqnarray*}
If $\Phi_1, \Phi_2 \in \mathcal{M}$ and $\Phi_1=\Phi_2$, then $S_{\Phi_1}=S_{\Phi_2}$ and, consequently, $\mathcal{A}(a,\zeta, \Phi_1)=\mathcal{A}(a,\zeta, \Phi_2),$
and consequently, $\mathcal{F}$ is well-defined.
\end{proof}

\medskip

In the following lemma we show that $\mathcal{F}:\; (\mathcal{M},d) \to (\mathcal{M},d)$ is a contraction.

\begin{lemma}
The map $\mathcal{F}:\; (\mathcal{M},d) \to (\mathcal{M},d)$ defined in  Lemma-\ref{define_L} is a contraction.
\end{lemma}

\begin{proof}
From \eqref{eqn_G}, we have

\begin{subequations}
\begin{eqnarray*}
\mathcal{A}(a,\zeta, \Phi) &=&  \alpha(a) (1-\rho_{\Phi}(a,\zeta)) S_0(a) e^{- \alpha(a) \int_0^{\zeta} (1-\rho_{\Phi}(a,\xi)) \big( \int_0^ {A^{\dagger}}  V(y) \Phi(y,\xi) dy \big) d \xi} \int_0^ {A^{\dagger}} V(y) \Phi(y,\zeta) dy \nonumber \\
 && - \int_0^{\zeta} \bigg[ \big( r(a,\zeta-\xi) + d(a,\zeta-\xi) \big)\alpha(a) (1-\rho_{\Phi}(a,\xi)) \nonumber \\
 && S_0(a) e^{- \alpha(a) \int_0^{\xi} (1-\rho_{\Phi}(a,\nu))\big( \int_0^ {A^{\dagger}}  V(y) \Phi(y,\nu) dy \big) d \nu} \int_0^ {A^{\dagger}}  V(y) \Phi(y,\xi) dy \bigg]  d \xi .
\end{eqnarray*}
\end{subequations}
 Let
$$\widehat{\Phi}(a,\xi)= (1-\rho_\Phi(a,\xi))\int_0^{A^{\dagger}} V(y) \Phi(y,\xi)dy.$$
Then the above expression for $\mathcal{A}(a,\zeta, \Phi)$ can be written as follows:

\begin{subequations}
\begin{eqnarray*}
\mathcal{A}(a,\zeta, \Phi) &=&  \alpha(a) S_0(a) e^{- \alpha(a) \int_0^{\zeta} \widehat{\Phi}(a,\xi) d \xi} \widehat{\Phi}(a,\zeta) \nonumber - \int_0^{\zeta} \bigg[ \big( r(a,\zeta-\xi) + d(a,\zeta-\xi) \big)\alpha(a)  \nonumber \\
 && S_0(a) e^{- \alpha(a) \int_0^{\xi} \widehat{\Phi}(a,\nu) d \nu} \widehat{\Phi}(a,\xi) \bigg]  d \xi .
\end{eqnarray*}
\end{subequations}
For any $\Phi_1(a,t), \; \Phi_2(a,t) \in \mathcal{M}$,
\begin{subequations}
\begin{eqnarray*}
\big| \mathcal{F}(\Phi_1)(a,t)- \mathcal{F}(\Phi_2)(a,t) \big| &\leq& \int_0^t \big| \mathcal{A}(a,\zeta, \Phi_1)-\mathcal{A}(a,\zeta, \Phi_2) \big| d \zeta.
\end{eqnarray*}
\end{subequations}
Next, 
\begin{subequations}
\begin{eqnarray*}
\big| \mathcal{A}(a,\zeta, \Phi_1)-\mathcal{A}(a,\zeta, \Phi_2) \big| &=& \alpha(a) S_0(a) \bigg| \bigg[ e^{-\alpha(a) \int_0^{\zeta} \widehat{\Phi}_1(a,\xi) d \xi} \widehat{\Phi}_1(a, \zeta)\\
&& -\int_0^{\zeta} (r(a,\zeta-\xi)+d(a,\zeta-\xi))e^{-\alpha(a) \int_0^{\xi} \widehat{\Phi}_1(a,\nu) d \nu} \widehat{\Phi}_1(a,\xi) d\xi \bigg]\\
&& - \bigg[ e^{-\alpha(a) \int_0^{\zeta} \widehat{\Phi}_2(a,\xi) d \xi} \widehat{\Phi}_2( x,\zeta)\\
&& -\int_0^{\zeta} (r(a,\zeta-\xi)+d(a,\zeta-\xi))e^{-\alpha(a) \int_0^{\xi} \widehat{\Phi}_2(a,\nu) d \nu} \widehat{\Phi}_2(a,\xi) d\xi \bigg]  \bigg|,
\end{eqnarray*}
\end{subequations}
where
$$\widehat{\Phi}_j(a,\xi)= (1-\rho_{\Phi_j}(a,\xi))\int_0^{A^{\dagger}} V(y) \Phi_j(y,\xi)dy, \;\;\; j=1,2.$$
From  equation \eqref{eqn_rho_2}, we observe that if $\rho_0(a) \geq0$, then, $(1-\rho_{\Phi_j}(a,\xi))$ always non-negative, and hence $\widehat{\Phi}_j(a,\xi)$ is always non-negative, for $j=1,\,2$. To simplify  forthcoming mathematical expressions, we use the notation
$$\Delta(a,.)\,=\,\widehat{\Phi}_1( x,.) -\widehat{\Phi}_2(a,.), $$
with `.' stands for $\xi$, $\zeta$ and $\nu$ as per necessity.
\begin{eqnarray*}
\big| \mathcal{A}(a,\zeta, \Phi_1)-\mathcal{A}(a,\zeta, \Phi_2) \big| &=& \alpha(a) S_0(a) \bigg| e^{-\alpha(a) \int_0^{\zeta} \widehat{\Phi}_1(a,\xi) d \xi} \Delta(a,\zeta)
\\
&& -\int_0^{\zeta} (r(a,\zeta-\xi)+d(a,\zeta-\xi))e^{-\alpha(a) \int_0^{\xi} \widehat{\Phi}_1(a,\nu) d \nu}\Delta(a,\xi) d\xi\\
&& + \big(e^{-\alpha(a) \int_0^{\zeta} \widehat{\Phi}_1(a,\xi) d \xi}- e^{-\alpha(a) \int_0^{\zeta} \widehat{\Phi}_2(a,\xi) d \xi} \big) \widehat{\Phi}_2( x,\zeta)+\int_0^{\zeta} \bigg(r(a,\zeta-\xi)\\
&& +d(a,\zeta-\xi)\bigg) \bigg( e^{-\alpha(a) \int_0^{\xi} \widehat{\Phi}_2(a,\nu) d \nu} -e^{-\alpha(a) \int_0^{\xi} \widehat{\Phi}_1(a,\nu) d \nu} \bigg) \widehat{\Phi}_2(a,\xi) d\xi  \bigg|.
\end{eqnarray*}
For any $p,q \geq 0,$ we have
$$\big|e^{-p}-e^{-q} \big| \leq \big| p-q\big|, \;\; \big| e^{-p} \big| \leq 1.$$
Using the above inequalities we get
\begin{subequations}
\begin{eqnarray*}
\big| \mathcal{A}(a,\zeta, \Phi_1)-\mathcal{A}(a,\zeta, \Phi_2) \big| &\leq& \alpha(a) S_0(a)\bigg( \big|\Delta(a,\zeta)\big|
+\int_0^{\zeta} (r(a,\zeta-\xi)+d(a,\zeta-\xi)) \big|\Delta(a,\xi)\big|d\xi\\
&& + \alpha(a)  \widehat{\Phi}_2(a,\zeta)  \int_0^{\zeta} \big|\Delta(a,\xi)\big|d\xi\\
&&+\int_0^{\zeta} (r(a,\zeta-\xi)+d(a,\zeta-\xi)) \big( \alpha(a)  \int_0^{\xi}\big|\Delta(a,\nu)\big|  d \nu \big) \widehat{\Phi}_2(a,\xi) d\xi  \bigg). 
\end{eqnarray*}
\end{subequations}
Since, $\Phi_j(a,\xi) \leq S_0(a)+I_0(a)$ and $0 \leq (1-\rho_{\Phi_j}(a,\xi)) \leq 1$, we can write
$$  0 \leq \widehat{\Phi}_j(a,\xi)  \leq M,  \;\;\;\;j=1,2,$$
where
$$M= (S_0(a)+I_0(a)) \int_0^{A^{\dagger}} V(y) dy.$$
Hence,
\begin{subequations}
\begin{eqnarray*}
\big|\Delta(a,\zeta)\big|&=& \Big|(1-\rho_{\Phi_1}(a,\zeta))\int_0^{A^{\dagger}} V(y) \Phi_1(y,\zeta)dy-(1-\rho_{\Phi_2}(a,\zeta))\int_0^{A^{\dagger}} V(y) \Phi_2(y,\zeta)dy \Big|\\
&=& \Big|(1-\rho_{\Phi_1}(a,\zeta))\int_0^{A^{\dagger}} V(y) (\Phi_1(y,\zeta)-\Phi_2(y,\zeta))dy+ \\
&& \int_0^{A^{\dagger}} V(y) \Phi_2(y,\zeta)dy (\rho_{\Phi_2}(a,\zeta)-\rho_{\Phi_1}(a,\zeta)) \Big|.
\end{eqnarray*}
\end{subequations}
Let $M_1= \int_0^{A^{\dagger}} V(y) dy$. Using the estimate $0 \leq (1-\rho_{\Phi_1}(a,\xi)) \leq 1$, we get from the above equation
\begin{eqnarray*}
\big|\Delta(a,\zeta)\big|&\leq& M_1 e^{\delta \zeta} d(\Phi_1,\Phi_2) + M \big|\rho_{\Phi_1}(a,\zeta)-\rho_{\Phi_2}(a,\zeta) \big|.
\end{eqnarray*}
From relation \eqref{eqn_rho_2} we obtain the following inequalities:
\begin{subequations}
\begin{eqnarray*}
\big|\rho_{\Phi_1}(a,\zeta)-\rho_{\Phi_2}(a,\zeta) \big|
&\leq& \Big| e^{- \int_0^{\zeta} \mathcal{P}(m_{\Phi_1}(a,\xi)) d \xi}- e^{- \int_0^{\zeta} \mathcal{P}(m_{\Phi_2}(a,\xi)) d \xi} \Big|\\
&\leq& \Big|\int_0^{\zeta} \mathcal{P}(m_{\Phi_1}(a,\xi)) d \xi -\int_0^{\zeta} \mathcal{P}(m_{\Phi_2}(a,\xi)) d \xi \Big| \\
&\leq & \int_0^{\zeta} \Big| \mathcal{P}(m_{\Phi_1}(a,\xi)) - \mathcal{P}(m_{\Phi_2}(a,\xi))  \Big|  d \xi.
\end{eqnarray*}
\end{subequations}
We assume that the function $\mathcal{P} \circ m(\Phi)=\mathcal{P}( m(\Phi))$ is a Lipschitz continuous function of $\Phi$, that is, 
\begin{equation}\label{lipschitz_11}
 d\big((\mathcal{P} \circ m)(\Phi_1), (\mathcal{P} \circ m) (\Phi_2)\big) \leq K d(\Phi_1, \Phi_2),
\end{equation}
where $K \geq 0$ is a constant.
This implies the estimate
\begin{equation}\label{lipschitz_2}
 \Big| \mathcal{P}(m_{\Phi_1}(a,\xi)) - \mathcal{P}(m_{\Phi_2}(a,\xi))  \Big| =  \Big| (\mathcal{P} \circ m)(\Phi_1(a,\xi)) - (\mathcal{P} \circ m)(\Phi_2(a,\xi))  \Big| \leq K e^{\delta \xi} d( \Phi_1, \Phi_2).
\end{equation}
Using \eqref{lipschitz_2}, we obtain

\begin{subequations}
\begin{eqnarray*}
\big|\rho_{\Phi_1}(a,\zeta)-\rho_{\Phi_2}(a,\zeta) \big| &\leq& K d(\Phi_1, \Phi_2) \int_0^{\zeta} e^{\delta \xi} d \xi =
 K d(\Phi_1, \Phi_2) \frac{e^{\delta \zeta}}{\delta},
\end{eqnarray*}
\end{subequations}
and, consequently, we get
$$ \big| \widehat{\Phi}_1( x,\zeta) -\widehat{\Phi}_2( x,\zeta) \big| \leq M_3 e^{\delta \zeta} d(\Phi_1, \Phi_2),$$
where $M_3=M_1+\frac{M K}{\delta}$.
From these relations and condition \eqref{cond}, we can write
\begin{subequations}
\begin{eqnarray*}
\big| \mathcal{A}(a,\zeta, \Phi_1)-\mathcal{A}(a,\zeta, \Phi_2) \big| &\leq& \alpha(a) S_0(a) M_3 d(\Phi_1,\Phi_2) \bigg[ e^{\delta \zeta} +\int_0^{\zeta}e^{\delta \xi}  d\xi \\
&& + \alpha(a) M  \int_0^{\zeta} e^{\delta \xi} d \xi +\alpha(a) M \int_0^{\zeta}    \int_0^{\xi} e^{\delta \nu} d \nu   d\xi  \bigg]\\\\
&\leq& \alpha(a) S_0(a) M_3 d(\Phi_1,\Phi_2) \bigg[ e^{\delta \zeta} +\frac{1}{\delta}(e^{\delta \zeta} -1)\\
&& + \alpha(a) M  \frac{1}{\delta}(e^{\delta \zeta} -1) +\alpha(a) M \frac{1}{\delta^2}(e^{\delta \zeta} -1) \bigg]\\\\
&\leq& \alpha(a) S_0(a) M_3 d(\Phi_1,\Phi_2) \bigg[ e^{\delta \zeta} \bigg( 1+\frac{1}{\delta} +\frac{\alpha(a) M}{\delta} +\frac{\alpha(a) M}{\delta^2} \bigg)\bigg].\\
\end{eqnarray*}
\end{subequations}
Therefore,

\begin{subequations}
\begin{eqnarray*}
\big| \mathcal{F}(\Phi_1)(a,t)- \mathcal{F}(\Phi_2)(a,t) \big| &\leq& \alpha(a) S_0(a) M_3 d(\Phi_1,\Phi_2)  \bigg( 1+\frac{1+\alpha(a) M}{\delta} +\frac{\alpha(a) M}{\delta^2} \bigg) \int_0^t e^{\delta \zeta} d \zeta\\\\
 &=& \alpha(a) S_0(a) M_3 d(\Phi_1,\Phi_2)  \bigg( 1+\frac{1+\alpha(a) M}{\delta} +\frac{\alpha(a) M}{\delta^2} \bigg) \frac{e^{\delta t}-1}{\delta},\\
\end{eqnarray*}
\end{subequations}
which implies

\begin{subequations}
\begin{eqnarray*}
e^{-\delta t} \big| \mathcal{F}(\Phi_1)(a,t)- \mathcal{F}(\Phi_2)(a,t) \big| &\leq& \alpha(a) S_0(a) M_3 d(\Phi_1,\Phi_2)  \bigg( \frac{1}{\delta}+ \frac{1+\alpha(a) M}{\delta^2}  +\frac{\alpha(a) M}{\delta^3} \bigg) \\\\
&=& \bigg[ \alpha(a) S_0(a) M_3 \bigg( \frac{1}{\delta}+ \frac{1+\alpha(a) M}{\delta^2}  +\frac{\alpha(a) M}{\delta^3} \bigg)  \bigg] d(\Phi_1,\Phi_2).
\end{eqnarray*}
\end{subequations}
By the positivity and continuity of the functions $\alpha(a)$, $S_0(a)$ on $[0, A^{\dagger}]$, there exists positive constants $M_{\alpha}$ and $M_s$ such that
$$\alpha(a) \leq M_{\alpha}\;\;\;\;\;\; S_0(a) \leq M_s,\;\;$$
for all $x \in [0,A^{\dagger}].$
Then we have
\begin{subequations}
\begin{eqnarray*}
e^{-\delta t} \big| \mathcal{F}(\Phi_1)(a,t)- \mathcal{F}(\Phi_2)(a,t) \big| &\leq& 
 \bigg[ M_{\alpha}M_s M_3 \bigg( \frac{1}{\delta}+ \frac{1+M_{\alpha} M}{\delta^2}  +\frac{M_{\alpha} M}{\delta^3} \bigg) \bigg] d(\Phi_1,\Phi_2).
\end{eqnarray*}
\end{subequations}
Now in both sides, we take supremum and obtain
\begin{eqnarray*}
d(\mathcal{F}(\Phi_1), \mathcal{F}(\Phi_2)) &\leq& 
 \bigg[ M_{\alpha}M_s M_3 \bigg( \frac{1}{\delta}+ \frac{1+M_{\alpha} M}{\delta^2}  +\frac{M_{\alpha} M}{\delta^3} \bigg) \bigg] d(\Phi_1,\Phi_2).
\end{eqnarray*}
Substituting $M_3=M_1+\frac{M KK'}{\delta}$, we obtain the inequality
\begin{subequations}
\begin{eqnarray*}
d(\mathcal{F}(\Phi_1), \mathcal{F}(\Phi_2)) &\leq& 
 \bigg[ M_{\alpha}M_s \left(M_1+\frac{M KK'}{\delta}\right) \left( \frac{1}{\delta}+ \frac{1+M_{\alpha} M}{\delta^2}  +\frac{M_{\alpha} M}{\delta^3} \right) \bigg] d(\Phi_1,\Phi_2).
\end{eqnarray*}
\end{subequations}
Now, we can fix $\delta>0$ so large that,
$$  M_{\alpha}M_s \left(M_1+\frac{M KK'}{\delta}\right) \left( \frac{1}{\delta}+ \frac{1+M_{\alpha} M}{\delta^2}  +\frac{M_{\alpha} M}{\delta^3} \right) < 1 .$$
Therefore, The map $\mathcal{F}$ defined on $(\mathcal{M},d)$, is a contraction map.
\end{proof}

From the last lemma and  Theorem \ref{Banach_theorem} it follows that the map $\mathcal{F}:\; (\mathcal{M},d) \to (\mathcal{M},d)$ has a unique fixed point $ I_u \in \mathcal{M} \subset C\big( [0, A^{\dagger}] \times [0, \mathcal{T}],\; \mathbb{R} \big)$ satisfying

$$I_u(a,t) = I_0(a) + \int_0^t \mathcal{A}(a,\zeta, I_u) d \zeta,$$
where $\mathcal{A}(a,\zeta,I_u)$ is given in \eqref{eqn_G}. Also, we observe that $\mathcal{A}(a,\zeta,\Phi)$ is a continuous function. Therefore, the derivative $\frac{\partial I_u(a,t)}{ \partial t}$ exists, and consequently, the existence of uniqueness solution of the system \eqref{s4}-\eqref{s5} is proved.


\section{Age-independent model}\label{age_independent_model}

We consider in this section a particular case of model \eqref{s4}-\eqref{s5} with age-independent parameters. Set
$$\widehat{S}(t)=\int_0^{A^{\dagger}}S(a,t)dx,\;\;\widehat{I}(t)=\int_0^{A^{\dagger}} I(a,t)dx,\;\;\widehat{R}(t)=\int_0^{A^{\dagger}} R(a,t)dx,\;\;\widehat{D}(t)=\int_0^{A^{\dagger}} D(a,t)dx.$$
Let the information index be age-independent, i.e., $m(t)\equiv m(\widehat{I}(t))$. Denote by $\widehat{\rho}(t)$ the proportion of people vaccinated at time $t$. Let $\widehat{P}_I(m)$ denote the age-independent information. Then the equation for the vaccination becomes as follows: 
$$   \frac{d \widehat{\rho}(t)}{d t} = \widehat{\mathcal{P}}_I(m) (1-\widehat{\rho}(a,t)).$$
Set $\alpha=\alpha(a)$, $W=V(a)$,  $\beta=\alpha W$, $r(t)=r(a,t)$ and $d(t)=d(a,t)$. Now, integrating the system \eqref{s4}-\eqref{s5} from $0$ to ${A^{\dagger}}$ with respect to x, we obtain the reduced system (hat is omitted):

\begin{equation}\label{hpp11}
  \frac{d S(t)}{d t} = - \beta S(t)(1-\rho(t)) I(t) \;\;( \equiv -J(t)),
\end{equation}
\begin{equation}\label{hpp12}
   \frac{d I(t)}{d t}= \beta S(t)(1-\rho(t)) I(t)  - \int_0^t r(t-\zeta) J(\zeta) d \zeta -\int_0^t d(t-\zeta) J(\zeta) d \zeta,
\end{equation}
\begin{equation}\label{hpp13}
   \frac{d R(t)}{d t} =\int_0^t r(t-\zeta) J(\zeta) d \zeta,
\end{equation}
\begin{equation}\label{hpp14}
   \frac{d D(t)}{d t} = \int_0^t d(t-\zeta) J(\zeta) d \zeta,
\end{equation}
\begin{equation}\label{hpp15}
   \frac{d \rho(t)}{d t} = \mathcal{P}(m(t)) (1-\rho(t)).
 \end{equation}
 This reduced age-independent model \eqref{hpp11}-\eqref{hpp15} is simpler than the original age-distributed model \eqref{s4}-\eqref{s5}, but it still captures the effect of information-dependent behavior of vaccine uptake. In the next subsection we find a bound for the final size of epidemic for this model.


\subsection{Final size of epidemic}

Let $S_f$ be the final size of the susceptible compartment for asymptotically large time, $\lim_{t \rightarrow \infty} S (t) = S_f$.
Integrating equation \eqref{hpp11} from $0$ to $\infty$ with respect to $t$, we get
\begin{equation}\label{final_size_eqn_1}
S_f - S_0 = - \int_0^{\infty} J(\zeta) d\zeta.
\end{equation}
On the other hand, integrating equation
$$\frac{d S(t)}{d t} = - \beta S(t)(1-\rho(t)) I(t),$$
we obtain
\begin{equation}\label{final_size_eqn_2}
\ln \frac{S_f}{S_0} = - \beta \int_0^{\infty} (1-\rho(\zeta)) I(\zeta) d \zeta.
\end{equation}
Furthermore,
\begin{subequations}
\begin{eqnarray*}
I(t)&=& \int_0^ t J(\zeta) d \zeta -R(t) -D(t) 
= \int_0^{t} J(\zeta) d\zeta -\int_0^t \int_0^{\zeta} \left( r(\zeta-\xi)+d(\zeta-\xi) \right) J(\xi) d \xi d \zeta\\
&=& \int_0^{t} \bigg[ J(\zeta)  - \int_0^{\zeta} \left( r(\zeta-\xi)+d(\zeta-\xi) \right) J(\xi) d \xi \bigg] d \zeta.
\end{eqnarray*}
\end{subequations}
Multiplying this equation by $(1-\rho(t))$ and integrating from $0$ to $\infty$ with respect to $t$, we obtain

\begin{subequations}
\begin{eqnarray*}
\int_0^{\infty} (1-\rho(t)) I(t) dt &=& \int_0^{\infty} (1-\rho(t)) \bigg( \int_0^{t} \bigg[ J(\zeta)  - \int_0^{\zeta} \left( r(\zeta-\xi)+d(\zeta-\xi) \right) J(\xi) d \xi \bigg] d \zeta \bigg) dt.
\end{eqnarray*}
\end{subequations}
We change the order of integration with respect to $\zeta$ and $t$ to get

\begin{subequations}
\begin{eqnarray*}
\int_0^{\infty} (1-\rho(t)) I(t) dt &=& \int_0^{\infty} \int_{\zeta}^{\infty} (1-\rho(t))   \bigg[ J(\zeta)  - \int_0^{\zeta} \left( r(\zeta-\xi)+d(\zeta-\xi) \right) J(\xi) d \xi \bigg]  dt d \zeta.
\end{eqnarray*}
\end{subequations}
Let $V(\zeta)=\int_{\zeta}^{\infty} (1-\rho(t)) dt$. Then we can write
\begin{eqnarray*}
\int_0^{\infty} (1-\rho(t)) I(t) dt& = &\int_0^{\infty} J(\zeta) V(\zeta) d \zeta  - \int_0^{\infty} \int_0^{\zeta} \left( r(\zeta-\xi)+d(\zeta-\xi) \right) J(\xi) V(\zeta) d \xi d \zeta  \\
& = &\int_0^{\infty} J(\xi) V(\xi) d \xi  - \int_0^{\infty} \int_{\xi}^{\infty} \left( r(\zeta-\xi)+d(\zeta-\xi) \right) V(\zeta) d \zeta J(\xi)  d \xi  \\
& =&\int_0^{\infty} \bigg[  V(\xi)  -  \int_{\xi}^{\infty} \left( r(\zeta-\xi)+d(\zeta-\xi) \right) V(\zeta) d \zeta \bigg] J(\xi)  d \xi.
\end{eqnarray*}
Furthermore,
\begin{eqnarray*}
\int_{\xi}^{\infty} \left( r(\zeta-\xi)+d(\zeta-\xi) \right) V(\zeta) d \zeta &=&
 \int_{\xi}^{\infty} \int_{\zeta}^{\infty} \left( r(\zeta-\xi)+d(\zeta-\xi) \right) (1-\rho(t))dt d\zeta  \\
 & = & \int_{\xi}^{\infty} \bigg( \int_{\xi}^{t} \left( r(\zeta-\xi)+d(\zeta-\xi) \right) d\zeta \bigg) (1-\rho(t))dt, 
\end{eqnarray*}
with  $V(\xi)=\int_{\xi}^{\infty} (1-\rho(t))dt$ gives,
\begin{eqnarray*}
V(\xi)  -  \int_{\xi}^{\infty} \left( r(\zeta-\xi)+d(\zeta-\xi) \right) V(\zeta) d \zeta\,=\,\int_{\xi}^{\infty} \bigg(1  - \int_{\xi}^{t} \left( r(\zeta-\xi)+d(\zeta-\xi) \right) d\zeta \bigg) (1-\rho(t))dt .
\end{eqnarray*}
Inequality
$$ \bigg(1  - \int_{\xi}^{t} \left( r(\zeta-\xi)+d(\zeta-\xi) \right) d\zeta \bigg) <1 $$ 
follows from condition \eqref{cond}. For some particular choice of the functions $r(\zeta)$ and $d(\zeta)$, a more precise estimate can be obtained. We have
$$V(\xi)  -  \int_{\xi}^{\infty} \left( r(\zeta-\xi)+d(\zeta-\xi) \right) V(\zeta) d \zeta  \leq \int_{\xi}^{\infty} (1-\rho(t))dt . $$ 
Therefore,   
$$ \int_0^{\infty} (1-\rho(t)) I(t) dt \leq \int_0^{\infty} \int_{\xi}^{\infty} (1-\rho(t)) J(\xi) dt d \xi. $$
Integrating \eqref{hpp15} from $0$ to $t$, we get
$$ (1-\rho(t))= (1-\rho_0) e^{-\int_0^t \mathcal{P}(m(I(\zeta))) d \zeta} .$$
For simplicity of calculation we assume that 
$\mathcal{P}(m(I(\zeta))) \geq k$. This implies
$(1-\rho(t)) \leq (1-\rho_0) e^{-kt}$
and, consequently,

\begin{subequations}
\begin{eqnarray*}
 \int_0^{\infty} (1-\rho(t)) I(t) dt &\leq&  \int_0^{\infty} \bigg( \int_{\xi}^{\infty}  (1-\rho_0) e^{-kt} dt \bigg) J(\xi) d \xi 
  \leq \frac{(1-\rho_0)}{k} \int_0^{\infty} J(\xi) d \xi.
\end{eqnarray*}
\end{subequations}
From this inequality and \eqref{final_size_eqn_1}, \eqref{final_size_eqn_2}, we get the following estimate for the final size of epidemic:

\begin{equation}\label{final_size_F}
\frac{\ln \frac{S_0}{S_f}}{S_0 - S_f} \leq \frac{\beta (1-\rho_0)}{k}.
\end{equation}
Set $x=S_f/S_0$, $\Lambda= \beta (1-\rho_0) S_0/k$. Then  inequality (\ref{final_size_F}) can be written as follows:
\begin{equation}\label{final_size_F2}
\frac{\ln x}{x-1} \geq \Lambda.
\end{equation}
Denote $\mathcal{G}(a)=\frac{\ln x}{x-1}$. Clearly, $\mathcal{G}(a)$ is a strictly decreasing function in $(0, 1)$, $\lim_{x \to 0} \mathcal{G}(a) = \infty$ and $\lim_{x \to 1-} \mathcal{G}(a) =1$. Suppose that $\Lambda >1$. Then there exists a unique value $x^* \in (0, 1)$ such that $\mathcal{G}(a^*)=\Lambda$. If    $\mathcal{G}(a) \geq \Lambda$ for some $a$, then $x \leq x^*$. Note that 

$$ x^* \propto \frac{1}{\Lambda} = \frac{\beta (1-\rho_0) S_0}{k} . $$ 
Hence, the upper bound of final size of epidemic $x^*$ is proportional to $k$, the minimal probability of vaccination.


\section{Numerical simulations}
\label{simulation_results}

In this section we illustrate the effect of information-based vaccination decision on the epidemic progression with numerical simulations of the reduced age-independent \eqref{hpp11}-\eqref{hpp15}. We consider the distributed rate functions $r(t)$ and $d(t)$ as used in \cite{bmb},
$r(t)=p_0 f_1(t)$ and $d(t)=(1-p_0)f_2(t)$ where

$$f_1(t)=\frac{1}{b_1^{a_1} \Gamma (a_1)} t^{a_1-1} 
e^{-\frac{t}{b_1}} \;\;\; \text{and}\;\;\; f_2(t)=\frac{1}{b_2^{a_2} \Gamma (a_2)} t^{a_2-1} e^{-\frac{t}{b_2}},$$
with the values of parameters $a_1=8.06275$, $b_1=2.21407$, $a_2=6.00014$, $b_2=2.19887$ and $p_0=0.97$.
%
%
%
%
%
%

Let us consider the information index in the form \cite{d2007vaccinating}
$$m(t)=\int_0^t K_{time}(\zeta) I(t-\zeta) d \zeta,$$
where the kernel $K_{time}$ corresponds to the information about infection progression. We consider two types of kernels, with memory decay:
$$K_{time}(\zeta)=a e^{-a \zeta},$$
where $a >0$ is the decay rate, and with memory acquisition-decay:
$$K_{time}(\zeta)=\frac{b d}{d-b} ( e^{-b\zeta} - e^{-d \zeta}),$$
where $0<b<d$ (Fig.~\ref{different_K_1}, a). The exponentially fading kernel captures the fact that the information for vaccination decision
declines with time. The acquisition-decay kernel captures the fact of acquisition of information followed by a fading of memory.

Let us note that both kernels are normalized in such a way that their integrals from $0$ to infinity equal $1$. Large values of $a$ in the first kernel correspond to the short-term memory, while small $a$ to the long term memory. Similarly, parameters of the second kernel determine the corresponding memory distribution. We will see below how the memory distribution influences epidemic progression.

\begin{figure}[ht!]
\begin{center}

      \mbox{\subfigure[]{\includegraphics[scale=0.45]{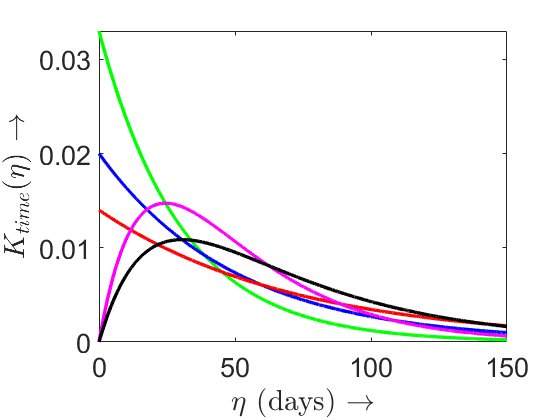}}
		   \subfigure[]{\includegraphics[scale=0.47]{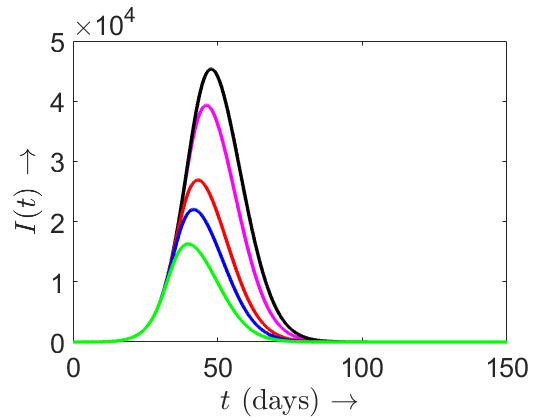}}}
		   	\mbox{\subfigure[]{\includegraphics[scale=0.45]{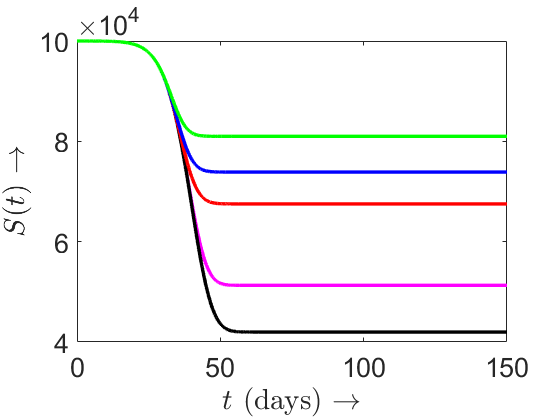}}}
\caption{(a) Graphical representation of the memory kernels $K_{time}(\zeta)=a e^{-a \zeta}$, with $a=0.033$ (green), $a=0.02$ (blue), $a=0.014$ (red) and $K_{time}(\zeta)=\frac{b d}{d-b} ( e^{-b\zeta} - e^{-d \zeta}),$ with $b=0.033, d=0.05$ (magenta); $b=0.02, d=0.05$ (black). (b) and (c) show the corresponding plots of $I(t)$ and $S(t)$, respectively. All other parameter values are given in the text.}
\label{different_K_1}
\end{center}
\end{figure}

To complete the formulation of the model, we define
the function $\mathcal{P}(m)$  as follows:

\begin{equation}\label{game_eqn_2}
 \mathcal{P}(m) = k+ Min \bigg\{c \big( m- m_0\big) \; Hev\big( m- m_0\big),\;1-k \bigg\} .
\end{equation}
Here $k$ is the minimal probability of vaccination independent of information, $k \in (0,1)$, $c$ is a non-negative constant, $Hev(\cdot)$ is the Heaviside function. This function equals $k$ for $m \leq 0$, it has a linear growth on some interval of positive values of $m$, and it equals $1$ for $m$ sufficiently large.

Numerical simulations of system \eqref{hpp11}-\eqref{hpp15} are presented in Fig.~\ref{different_K_1} (b, c)
for the values of parameters
$N=10^5$, $\beta=0.4/N$, $S(0)=N-1$, $I(0)=1$, $R(0)=0$, $D(0)=0$, $\rho(0)=0$, $k=0.02$, $c=10^{-4}$, $m_0=500$, and different values of parameters $a, b, d$ that determine the memory kernels. 

We note that the beginning of the epidemic outbreak does not depend on parameters and on the kernel type. However, further epidemic progression is different. 
For larger values of $a$, the maximal number of infected individuals and their total number is less than for small values of $a$. Similarly, the final size of susceptible individuals is larger. Hence, the short-term memory is more efficient to restrain epidemic progression. A similar conclusion holds for the second kernel.


\section{Discussion}

This work continues the investigation of epidemic models with time distributed recovery and death rates. Similar to the previous models \cite{mmnp_imep,mathematics_our,bmb}, the evolution of the number of susceptible individuals $S(t)$ and infected individuals $I(t)$ is determined through the number of newly infected individuals $J(t) = \beta S(t) I(t)/N$. This approach allows a more accurate understanding of the epidemic progression than conventional SIR model. An age-distributed model with distributed recovery and death rate was introduced in \cite{age_dependent_our1}, and the influence of vaccination waning was studied in \cite{mmnp_imep}. 

Vaccination campaign during COVID-19 pandemic was accompanied by wide spread hesitancy and resistance \cite{agarwal2021socioeconomic,dror2020vaccine,peretti2020future} related to the precipitation in the vaccine development and by the reporting of some adverse effects. This vaccination avoidance was reinforced by the social networks and various informal mass media \cite{della2021volatile,MBE_VV_AD}. Though it is difficult to estimate exactly the influence of available information on the vaccination behavior, we can reasonably assume that the willingness to vaccinate is proportional to the number of infected individuals \cite{baumgaertner2020risk,cerda2021willingness}. Moreover, information-based decision includes some memory effect taken into account through the memory kernel.

In order to describe the age-specific vaccination decision effects on epidemic progression, we propose an integro-differential system of equations for susceptible, infected, recovered/dead individuals, and  for the proportion of vaccinated individuals. To mathematically justify the proposed model, we proved the existence and uniqueness of positive solution with the fixed point theory. 

The model developed in this work is able to capture the effect of age-specific information-based vaccination decision. In the prediction regarding the epidemic, we need more detailed age-specific data to estimate the distributed parameters involved in the model. However, in a simpler case, for the age-independent model, we show the effect of information on epidemic progression. The kernel $K_{time}(\zeta)$ involved in the information index $m(t)$ reflects the information about the effect of vaccination on previous epidemic progression. 

We observe that the initial stage of the epidemic outbreak is similar for all kernels since the memory effect is not yet influential, and the epidemic progression is mainly determined by the initial conditions and the transmission dynamics of the disease. However, as the epidemic progresses and more information about the disease becomes available, the information-based vaccination decision can have a significant impact on the course of the epidemic. Short-term memory appears to be more efficient to restrain the epidemic burst because it allows for more responsive and adaptive vaccination decisions based on the most recent information about the disease. The question of age-distributed vaccination decision requires further investigations and availability of relevant data. In the case of COVID-19 epidemic, young age groups provide, at the same time, higher rate of epidemic transmission and stronger vaccination resistance due to lower mortality rate.

\section*{Acknowledgements} Vitaly Volpert has been supported by the RUDN University Scientific Projects Grant System, project No 025141-2-174.


\section*{Appendix 1}

 \begin{lemma}
 Consider the metric space $\big( \mathcal{M}, d \big) $ with
 $$d(\Phi_1,\Phi_2)=\sup_{(a,t)\in [0,A^{\dagger}]\times[0,\mathcal{T}]} \bigg\{ e^{-\delta t} | \Phi_1(a,t)-\Phi_2(a,t)| \bigg\},$$
 where $\delta >0$ is a constant. Then $\big( \mathcal{M}, d \big) $ is complete.
 \end{lemma}
 
\begin{proof} 
Clearly, $\mathcal{M}$ is a closed subset of $C\big( [0, A^{\dagger}] \times [0, \mathcal{T}],\; \mathbb{R} \big)$. We know that, $C\big( [0, A^{\dagger}] \times [0, \mathcal{T}],\; \mathbb{R} \big)$ is complete with respect to the metric
 $$d_{\sup}(\Phi_1,\Phi_2)=\sup_{(a,t)\in [0,A^{\dagger}]\times[0,\mathcal{T}]} \bigg\{ | \Phi_1(a,t)-\Phi_2(a,t)| \bigg\}.$$
Hence, $( \mathcal{M}, d_{\sup})$ is also complete. \\
 Now we have
 $$e^{-\delta \mathcal{T}} d_{\sup}(\Phi_1,\Phi_2) \leq d(\Phi_1,\Phi_2) \leq d_{\sup}(\Phi_1,\Phi_2),$$
which shows that $d_{\sup}$ and $d$ are equivalent metrices, and consequently, $(\mathcal{M},d)$ is a complete metric space.
 \end{proof}

\end{document}